\documentclass[copyright,creativecommons,sharealike]{eptcs}
\providecommand{\event}{E. Formenti (Ed.): AUTOMATA \& JAC 2012}

\usepackage{amsmath}
\usepackage{amssymb}
\usepackage{tikz}
\usetikzlibrary{calc}
\usetikzlibrary{patterns}
\usetikzlibrary{decorations.pathreplacing}

\usepackage[thmmarks]{ntheorem}
\theoremseparator{.}
\newtheorem{theorem}{Theorem} 
\theorembodyfont{\upshape}
\newtheorem*{assertion}{Assertion}
\newtheorem{corollary}[theorem]{Corollary}
\newtheorem{definition}[theorem]{Definition}
\newtheorem{lemma}[theorem]{Lemma}

\newtheorem*{remark}{Remark}

\qedsymbol{\ensuremath{_\blacksquare}}

\theoremsymbol{\ensuremath{_\blacksquare}}
\theoremheaderfont{\itshape}
\newtheorem*{proof}{Proof}

\newcommand{\card}[1]{|#1|}
\newcommand{\abs}[1]{\lvert #1 \rvert}
\newcommand{\curr}{\mathit{curr}}
\newcommand{\define}[1]{\emph{#1}} 
\newcommand{\eg}{for example,\ } 
\newcommand{\es}{\emptyset}
\newcommand{\ie}{that is,\ } 
\newcommand{\old}{\mathit{old}}
\newcommand{\Time}{\mathit{time}}

\newcommand{\set}[1]{\{#1\}}

\def\steprel#1{\mathrel{\mathcal{#1}}}
\newcommand{\RDelta}{\mathrel\Delta}
\newcommand{\RDeltaone}{\mathrel{\Delta_1}}
\newcommand{\RGamma}{\mathrel\Gamma}
\newcommand{\RGammaone}{\mathrel{\Gamma_1}}
\newcommand{\stepR}{\steprel{R}}
\newcommand{\x}{\times}
\newcommand{\Z}{\mathbf{Z}}
\newcommand{\N}{\mathbf{N}}
\newcommand{\CA}{cellular automaton} 
\newcommand{\CAs}{cellular automata} 
\newcommand{\ACA}{asynchronous cellular automaton} 
\newcommand{\ACAs}{asynchronous cellular automata} 
\newcommand{\DCA}{deterministic cellular automaton} 
\newcommand{\DCAs}{deterministic cellular automata} 

\title{Phase Space Invertible Asynchronous Cellular Automata}
\author{Simon Wacker and Thomas Worsch
  \institute{Karlsruhe Institute of Technology\\
    Department of Informatics \\
    \email{simon.wacker@student.kit.edu}
    \qquad \email{worsch@kit.edu}
  }
}
\def\authorrunning{S.~Wacker and Th.~Worsch}
\def\titlerunning{Phase Space Invertible Asynchronous Cellular Automata}

\begin{document}

\maketitle

\begin{abstract}
  While for synchronous deterministic cellular automata there is an
  accepted definition of reversibility, the situation is less clear for
  asynchronous cellular automata. We first discuss a few possibilities
  and then investigate what we call \emph{phase space invertible}
  asynchronous cellular automata in more detail. We will show that for
  each Turing machine there is such a cellular automaton simulating it,
  and that it is decidable whether an asynchronous cellular automaton
  has this property or not, even in higher dimensions.
\end{abstract}

\section{Introduction}
\label{sec:intro}

For synchronous deterministic cellular automata the topic of
reversibility has gained a good deal of attention. Reversibility is
known to be decidable in one dimension \cite{Amoroso_1972_DPS_ar} and
undecidable in two and more dimensions \cite{Kari_1994_RSP_ar}.
Reversible synchronous \CAs\ are also computationally universal (see,
\eg \cite{Morita_1989_CUO_ar} for the one-dimensional case).

In the present paper we take a first look at the analogous questions for
\emph{asynchronous} \CAs. Hence the rest of this paper is organized as
follows: In section~\ref{sec:basics} we introduce some notation used
in this paper and two variants of asynchronicity.
In section~\ref{sec:invertibility} we discuss several possibilities
for the definition of ``reversibility'' of \ACAs.
In section~\ref{sec:completeness} we show that each Turing machine can
be simulated by a phase space invertible purely \ACA, and in
section~\ref{sec:decidability} that the property of being phase space
invertible is decidable for several variants of \ACAs.

This paper is based on the diploma thesis of the first author
\cite{Wacker_2012_RAZ_mt}.

\section{Basics}
\label{sec:basics}

\subsection{General Notation}

We write $\N$ for the set of natural numbers without $0$,
$\Z$ for the set of integers, $0 \in \Z^d$ for the $d$-tuple with each
component equal to zero, $B^A$ for the set of all total functions from
$A$ to $B$, and $2^M$ for the powerset of a set $M$. The cardinality of
a set $M$ is denoted as $\card{M}$. The restriction of a function
$f:A\to M$ to a subset $B\subseteq A$ of its domain is written $f|_B$.

In this paper we are interested in $d$-dimensional cellular automata

($d \in \N$) and the set of \define{cells} is usually denoted as $R$, \ie
$R=\Z^d$.  If the set of \define{states} of one cell is denoted as $Q$, the set
of all \define{(global) configurations} is $Q^R$.  A \define{neighborhood} is a
finite set $N=\{n_1,\dotsc,n_k\} \subseteq \Z^d$ of $d$-tuples of integers. A

\define{local configuration} is a mapping $\ell:N\to Q$; thus $Q^N$ is
the set of all local configurations. The local configuration $c_{i+N}$
observed by cell $i\in R$ in the global configuration $c \in Q^R$ is
defined as
\begin{align*}
  c_{i+N}: N &\to Q,\\
           n &\mapsto c(i+n).
\end{align*}
The behavior of each single cell of a deterministic \CA\ is described by
the \define{local transition function} $\delta: Q^N \to Q$.


\subsection{Asynchronous Updating Schemes}
\label{subsec:updating-schemes}

A local structure $(R,N,Q,\delta)$ of a deterministic \CA\ together with
a prescription how cells are updated induces a \define{global transition
relation} $\mathord{\stepR} \subseteq Q^R\x Q^R$ describing the possible
global steps which satisfies
\begin{equation*}
  c \stepR c'
  \Longrightarrow
  \forall i \in R: c'(i) \in \set{\delta(c_{i+N}), c(i)}.
\end{equation*}
We will use the same symbol for the related \define{global transition
function}
\begin{align*}
  \mathord{\stepR}: Q^R &\to 2^{Q^R}\text{,}\\
            c &\mapsto \{ c' \mid c\stepR c' \} \text{.}
\end{align*}
With this notation $c\stepR c'$ is equivalent to
$c' \in \mathord{\stepR}(c)$ and
both indicate that it is possible to reach global configuration
$c'\in Q^R$ in one step from global configuration $c\in Q^R$.

In a global step each cell has two possibilities: to be \define{active}
and make a state transition according to the rule or to be
\define{passive} and maintain its state. Restrictions made by
different updating schemes lead to different possible behaviors, \ie
different relations/functions $\stepR$, of \CAs.

Now, we will have a look at two different types of asynchronous
updating.

\paragraph{Purely Asynchronous Updating.}

This version of asynchronous updating has been considered for
many years now \cite{Nakamura_1974_ACA_ar,Golze_1978_ANC_ar}. In order
to distinguish it from the other form mentioned below we call it
\define{purely asynchronous} updating.

In each
global step there are no restrictions on whether a cell may be active
or passive. Thus in each step there is a subset $A \subseteq R$ of
active cells which make a transition, while the cells in the complement
$R \setminus A$ are passive and simply maintain their state.

Note that $A$ is is allowed to be empty. The additional requirement
$A\not=\emptyset$ might look irrelevant, but a closer look at the
constructions and theorems reveals that it would render the main
results wrong (see for example the remark after
Lemma~\ref{lem:fully-inversion-characterized}).

Given a \DCA\ and a set $A\subseteq R$ we define the function
\begin{align*}
  \Delta_A: Q^R &\to Q^R,\\
           \forall i \in R\colon\text{\qquad}
           \Delta_A(c)(i) &= \begin{cases}
                     \delta(c_{i+N}) & \text{ if $i\in A$,}\\
                     c(i)            & \text{ if $i\notin A$}.\\
                   \end{cases}
\end{align*}
Synchronous updating is described by $\RDelta_R$. The union of all
$\RDelta_A$, interpreted as relations, is the general step relation for
purely asynchronous updating, for which we will write
\[ \Delta =\bigcup_{A\subseteq R} \RDelta_A \text{.} \]
Note that for each global configuration $c \in Q^R$ holds
\[ \Delta(c) = \set{\Delta_A(c) \mid A \subseteq R} \text{.} \]

\paragraph{Fully Asynchronous Updating.}

In the fully asynchronous updating scheme it is required that in each
global step exactly one cell is active. Using the notation from above
one may say that one only looks at the relations $\RDelta_{\{i\}}$
where the set of active cells is a singleton. For the union of these
relations we will write
\[ \Delta_1 =\bigcup_{i\in R} \RDelta_{\{i\}} \text{.} \]
Even for relatively simple \DCAs, \eg the elementary \DCAs\ or
two-dimensional minority, the analysis of their behavior under fully
asynchronous updating is surprisingly ``non-simple''
\cite{DBLP:conf/acri/FatesG08,DBLP:journals/tcs/RegnaultST09,Lee_2004_AGL_ar}.

To distinguish the global step relations of two purely or fully
asynchronous \CAs\ $C$ and $G$ we will use $\RDelta_A$, $\RDelta$, and
$\RDeltaone$ for the respective relations of $C$ and $\RGamma_A$,
$\RGamma$, and $\RGammaone$ for the respective relations of $G$.

\section{Which Definition of Invertibility?}
\label{sec:invertibility}

For the global transition function $\Delta$ of a synchronous
deterministic \CA\ $C$ the following conditions are equivalent:
\begin{description}
\item[\textmd{R1}] Each global configuration $c\in Q^R$ has exactly
  one predecessor under $\Delta$.
\item[\textmd{R2}] Each global configuration $c\in Q^R$ has at most
  one predecessor under $\Delta$.
\item[\textmd{R3}] The inverse of the transition graph of $C$,
  \ie the direction of each transition is reversed,
  is the
  transition graph of a(nother) synchronous deterministic \CA\ $C'$.
\end{description}
For asynchronous \CAs\ these are really different conditions. In that
case, obviously, condition R1 still implies R2 as well as R3. But the
reverse implications do not hold. In order to clarify this, we first
show that there is essentially only one asynchronous \CA\ satisfying R1
or R2: The identity.

\begin{lemma}
  \label{lem:nontriv-2-pred}
  If a purely or fully asynchronous \CA\ $C=(R,N,Q,\delta)$ has a local
  transition function $\delta$ which is non-trivial in the sense that
  $\delta(\ell) \not= \ell(0)$ holds for at least one local
  configuration $\ell$, then there are two different global
  configurations $\hat{c}$ and $\check{c}$ and two singleton sets of active cells
  $\set{\hat{a}}$ and $\set{\check{a}}$ such that
  $\Delta_{\set{\hat{a}}}(\hat{c})=\Delta_{\set{\check{a}}}(\check{c})$.
  (Without loss of generality, we assume that $0$ is in the
  neighborhood.)
\end{lemma}
In other words, there is a global configuration which has two different
predecessors.

Lemma~\ref{lem:nontriv-2-pred} shows that the requirement of R1 or R2
leaves the trivial identity as the only asynchronous \CA.
Later, in section~\ref{sec:completeness} and
section~\ref{sec:elementary}, we will see, that there are non-trivial
\ACAs\ satisfying R3. Therefore, for asynchronous \CAs, R3 does not
imply R1 or R2.

\begin{proof}[of Lemma~\ref{lem:nontriv-2-pred}]
  Let $\ell \in Q^N$ be a local configuration with $q'=\delta(\ell) \not=
  \ell(0)=q$.

  Let $a \in R$ be a cell large enough such that the neighborhoods $(-a)+N$ and
  $a+N$ are disjoint. We consider a global configuration
  $c\in Q^R$ in which cell $-a$ and cell $a$ both observe $\ell$ in their
  neighborhoods, \ie $c_{(-a)+N}=\ell$ and $c_{a+N}=\ell$, in particular
  $c(-a)=q$ and $c(a)=q$. Define two global configurations $\hat{c}$ and $\check{c}$
  which are identical to $c$ with the only exceptions $\hat{c}(-a)=q'$ and
  $\check{c}(a)=q'$ respectively. Since $q'\not= q$ the global configurations
  $\hat{c}$ and $\check{c}$ are different.

  \begin{figure}[ht]
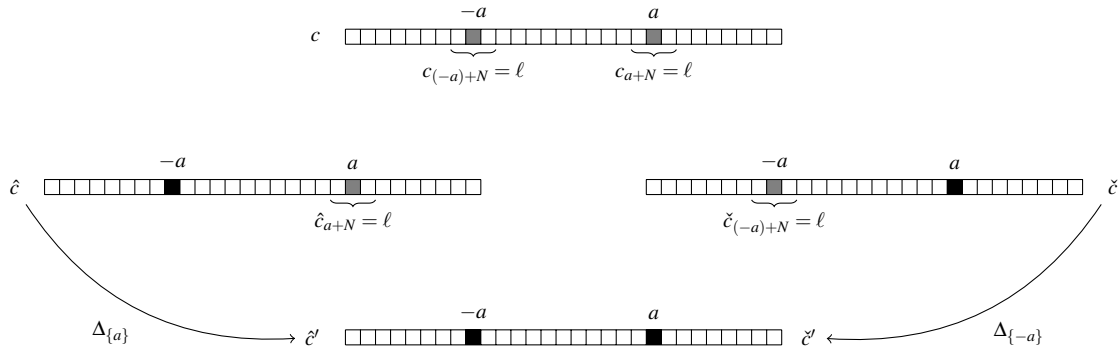

    \centering
    \tikz[scale=0.2, auto]{
      \path (14,0.5) node[anchor=east] (c1) {\scriptsize$c$};
      \path[fill=gray] (23,0) rectangle +(1,1) node [above,pos=0.5,yshift=3pt,scale=1] {\scriptsize$-a$}; 
      \draw[decoration={brace,mirror,raise=2pt}, decorate] (22,0) -- (25,0) node [below,pos=0.5,anchor=north,yshift=-3pt,scale=1] {\scriptsize$c_{(-a)+N} = \ell$};
      \path[fill=gray] (35,0) rectangle +(1,1) node [above,pos=0.5,yshift=3pt,scale=1] {\scriptsize$a$}; 
      \draw[decoration={brace,mirror,raise=2pt}, decorate] (34,0) -- (37,0) node [below,pos=0.5,anchor=north,yshift=-3pt,scale=1] {\scriptsize$c_{a+N} = \ell$};
      \draw[step=1] (15,0) grid +(29,1);
      \begin{scope}[shift={(-20,-10)}]
        \path (14,0.5) node[anchor=east] (c2) {\scriptsize$\hat{c}$};
        \path[fill] (23,0) rectangle +(1,1) node [above,pos=0.5,yshift=3pt,scale=1] {\scriptsize$-a$}; 
        \path[fill=gray] (35,0) rectangle +(1,1) node [above,pos=0.5,yshift=3pt,scale=1] {\scriptsize$a$}; 
        \draw[decoration={brace,mirror,raise=2pt}, decorate] (34,0) -- (37,0) node [below,pos=0.5,anchor=north,yshift=-3pt,scale=1] {\scriptsize$\hat{c}_{a+N} = \ell$};
        \draw[step=1] (15,0) grid +(29,1);
      \end{scope}
      \begin{scope}[shift={(20,-10)}]
        \path (47,0.5) node[anchor=east] (c3) {\scriptsize$\check{c}$};
        \path[fill=gray] (23,0) rectangle +(1,1) node [above,pos=0.5,yshift=3pt,scale=1] {\scriptsize$-a$}; 
        \draw[decoration={brace,mirror,raise=2pt}, decorate] (22,0) -- (25,0) node [below,pos=0.5,anchor=north,yshift=-3pt,scale=1] {\scriptsize$\check{c}_{(-a)+N} = \ell$};
        \path[fill] (35,0) rectangle +(1,1) node [above,pos=0.5,yshift=3pt,scale=1] {\scriptsize$a$}; 
        \draw[step=1] (15,0) grid +(29,1);
      \end{scope}
      \begin{scope}[shift={(0,-20)}]
        \path (14,0.5) node[anchor=east] (c4) {\scriptsize$\hat{c}'$};
        \path (47,0.5) node[anchor=east] (c5) {\scriptsize$\check{c}'$};
        \path[fill] (23,0) rectangle +(1,1) node [above,pos=0.5,yshift=3pt,scale=1] {\scriptsize$-a$}; 
        \path[fill] (35,0) rectangle +(1,1) node [above,pos=0.5,yshift=3pt,scale=1] {\scriptsize$a$}; 
        \draw[step=1] (15,0) grid +(29,1);
      \end{scope}
      \draw[->, bend right] (c2) to node[swap] {\scriptsize$\Delta_{\set{a}}$} (c4);
      \draw[->, bend left] (c3) to node {\scriptsize$\Delta_{\set{-a}}$} (c5);
    }
    \caption{Displayed are the relevant parts of the global
      configurations $c$, $\hat{c}$, $\check{c}$, $\hat{c}'$ and
      $\check{c}'$ in the one-dimensional case, \ie $R = \Z$, with
      neighborhood $N = \set{-1,0,1}$. Grey colored cells are in state
      $q$ and black colored cells in state $q'$. The state of white
      colored cells is not further specified. The solid directed
      lines indicate the two interesting global transitions.}
    \label{fig:lemma-1}
  \end{figure}

  But the two global configurations $\hat{c}'=\Delta_{\{a\}}(\hat{c})$ and
  $\check{c}'=\Delta_{-a}(\check{c})$ are the same (see Figure \ref{fig:lemma-1}).
  \begin{itemize}
  \item Each cell $i\notin \{-a,a\}$ is passive; therefore
    $\hat{c}'(i)=\hat{c}(i)=c(i)=\check{c}(i)=\check{c}'(i)$.
  \item $\hat{c}'(-a)= \hat{c}(-a) = q' = \delta(\ell) =
    \delta(\check{c}_{(-a)+N}) = \check{c}'(-a)$.
  \item $\hat{c}'(a)= \delta(\hat{c}_{a+N}) = \delta(\ell) = q' =
    \check{c}(a) = \check{c}'(a)$.
  \end{itemize}
\end{proof}
In a short presentation given at Automata 2011 Sarkar and Das
\cite{Sarkar_2011_RDA_ip} also consider some kind of
``reversibility of $1$-dimensional asynchronous cellular
automata''. We note that at least their setting is completely
different from ours: They only look at finite configurations (with
both, periodic and null, boundary conditions). And as far as we
understand their Definitions~1 and~2 \cite[p.~32]{Sarkar_2011_RDA_ip},
they call an \ACA, restricted to configurations of a fixed length,
reversible if each (finite) configuration has \emph{at least one}
predecessor. We do not pursue this line of thought.

Instead we use R3 as the guiding light. In order not to overload the
word ``reversible'' with too many meanings we will speak of phase
space invertibility, which we will often abbreviate as invertibility.

\begin{definition}
  \label{def:purely-invertible}
  A purely asynchronous \CA\ $C=(R,N,Q,\delta)$ is called
  \define{phase space invertible}, if there is a purely asynchronous
  \CA\ $G=(R,M,Q,\gamma)$ such that for each pair of global
  configurations $c, c'\in Q^R$ holds
  \[ c' \in\Delta(c) \iff c \in\Gamma(c') \text{,} \]
  which is equivalent to
  \begin{equation*}
    \exists A\subseteq R: c' = \Delta_A(c)
    \iff
    \exists A'\subseteq R: c = \Gamma_{A'}(c') \text{.}
  \end{equation*}
\end{definition}
We give a name to the set of cells on which two global configurations
differ in

\begin{definition}
  The \define{difference} $D_{c,c'}$ of two global configurations
  $c, c' \in Q^R$ is defined as
  \[ D_{c,c'} = \{ i\in R \mid c_i \not= c'_i \} \text{.} \]
  Note that $D_{c,c'} = D_{c',c}$.
\end{definition}
Since for purely asynchronous \CAs\ there are no restrictions on the
set of active cells, cells which do not change their state when active,
can be removed from the activity set without changing the outcome of a
transition.

This suggests the characterization of invertibility given in

\begin{lemma}
  \label{lem:purely-inversion-characterized}
  Two purely asynchronous \CAs\ $C=(R,N,Q,\delta)$ and
  $G=(R,M,Q,\gamma)$ are inverse to each other if, and only if, for each
  pair of global configurations $c, c'\in Q^R$ with $D_{c,c'} \neq \es$
  holds
  \[ c'=\Delta_{D_{c,c'}}(c) \iff c=\Gamma_{D_{c,c'}}(c') \text{.} \]
\end{lemma}

\begin{proof}\hfill
  \begin{itemize}
    \item[``only if'':] Let $C$ and $G$ be inverse to each other.
      Moreover let $c, c' \in Q^R$ be an arbitrary pair of global
      configurations with $D_{c,c'} \neq \es$.

      If $c' = \Delta_{D_{c,c'}}(c)$, then there exists a set of active
      cells $A'$ such that $c = \Gamma_{A'}(c')$ by premise. In this
      case $D_{c,c'} \subseteq A'$, since $c_i = c'_i$ for each cell
      $i \not\in A'$, and $c = \Gamma_{D_{c,c'}}(c')$, since
      $c_i = c'_i$ for each cell $i \in A' \setminus D_{c,c'}$.
      Therefore
      \begin{equation*}
        c' = \Delta_{D_{c,c'}}(c)
        \Longrightarrow
        c = \Gamma_{D_{c,c'}}(c')
        \text{.}
      \end{equation*}
      The other direction follows by symmetry.
    \item[``if'':] For each pair of global configurations
      $c, c'\in Q^R$ with $D_{c,c'} \neq \es$ let
      \[ c'=\Delta_{D_{c,c'}}(c)\iff c=\Gamma_{D_{c,c'}}(c')\text{.} \]
      Moreover let $c, c' \in Q^R$ be an arbitrary pair of global
      configurations.

      If there exists a set of active cells $A$ such that
      $c' = \Delta_A(c)$, then $c' = \Delta_{D_{c,c'}}(c)$, since
      $D_{c,c'} \subseteq A$ and $c'_i = c_i$ for each cell
      $i \in A \setminus D_{c,c'}$. In this case we get
      $c = \Gamma_{D_{c,c'}}(c')$ by premise if $D_{c,c'} \neq \es$ or
      by definition of $\Gamma_\es$ and equality $c = c'$ if
      $D_{c,c'} = \es$. Therefore
      \begin{equation*}
        \exists A\subseteq R: c' = \Delta_A(c)
        \Longrightarrow
        c = \Gamma_{D_{c,c'}}(c')
        \text{.}
      \end{equation*}
      The other direction follows by symmetry.
  \end{itemize}
\end{proof}
Analogously to definition~\ref{def:purely-invertible} we define
invertibility for fully asynchronous \CAs\ in

\begin{definition}
  A fully asynchronous \CA\ $C=(R,N,Q,\delta)$ is called
  \define{phase space invertible}, if there is a fully asynchronous
  \CA\ $G=(R,M,Q,\gamma)$ such that for each pair of global
  configurations $c, c'\in Q^R$ holds
  \[ c' \in\Delta_1(c) \iff c \in\Gamma_1(c') \text{,} \]
  which is equivalent to
  \begin{equation*}
    \exists a\in R: c' = \Delta_{\set{a}}(c)
    \iff
    \exists a'\in R: c = \Gamma_{\set{a'}}(c') \text{.}
  \end{equation*}
\end{definition}
We characterize inversion for fully asynchronous \CAs\ in

\begin{lemma}
  \label{lem:fully-inversion-characterized}
  Two fully asynchronous \CAs\ $C=(R,N,Q,\delta)$ and
  $G=(R,M,Q,\gamma)$ are inverse to each other if, and only if, for each
  pair of global configurations $c, c'\in Q^R$ with $D_{c,c'} = \set{a}$
  holds
  \[ c'=\Delta_{\set{a}}(c) \iff c=\Gamma_{\set{a}}(c') \]
  and for each global configuration $c \in Q^R$ holds
  \begin{equation*}
    \exists a\in R: c = \Delta_{\set{a}}(c)
    \iff
    \exists a'\in R: c = \Gamma_{\set{a'}}(c) \text{.}
  \end{equation*}
\end{lemma}
A comparison with the formulation in
lemma~\ref{lem:purely-inversion-characterized} shows a
complication. This is due to the fact that for fully
asynchronous \CAs\ one always needs at least one active cell.

\begin{proof}[of lemma~\ref{lem:fully-inversion-characterized}]
  Let $c, c' \in Q^R$ be an arbitrary pair of global configurations.
  There are three cases to consider:
  \begin{description}
    \item[Case 1:] $\card{D_{c,c'}} \geq 2$, \ie $c$ and $c'$ differ in
      more than one cell. Then $c' \notin \Delta_1(c)$ and
      $c \notin \Gamma_1(c')$, and hence
      \[ c' \in \Delta_1(c) \iff c \in \Gamma_1(c') \text{.} \]
    \item[Case 2:] $D_{c,c'} = \set{a}$, \ie $c$ and $c'$ differ in
      exactly one cell $a$. To reach $c'$ in one step from $c$ by
      $C$ or $c$ in one step from $c'$ by $G$ cell $a$ must be active, thus
      \begin{align*}
        c' \in \Delta_1(c) &\iff c' = \Delta_{\set{a}}(c) \text{ and}\\
        c \in \Gamma_1(c') &\iff c = \Gamma_{\set{a}}(c') \text{.}
      \end{align*}
    \item[Case 3:] $D_{c,c'} = \emptyset$, \ie $c = c'$. If $c'$
      can be reached from $c$ by $C$, \ie $c' \in \Delta_1(c)$, $c$ may
      not be reachable from $c'$ by $G$ or be reachable with a different
      active cell, and vice versa. Thus, a better characterization of
      this case is not as simple as in the other cases and is postponed
      until lemma~\ref{lem:restrict-active-cells-fully}.

      Note, that for purely asynchronous cellular automata this case is
      trivial: Simply choose the empty activity set.
  \end{description}
\end{proof}
Lemma~\ref{lem:purely-inversion-characterized} is only correct
because we allow the activity set to be empty. If we would not allow
this, we would additionally have to require for each global
configuration $c \in Q^R$ that
\begin{equation*}
  \exists a\in R: c = \Delta_{\set{a}}(c)
  \iff
  \exists a'\in R: c = \Gamma_{\set{a'}}(c) \text{.}
\end{equation*}
As we will see, this makes it much more difficult to prove completeness
and decidability.

\section{Turing Completeness}
\label{sec:completeness}

We will show that invertible purely asynchronous cellular automata are
computationally universal.

It is known that reversible deterministic synchronous \CAs\ are
computationally universal. For one-dimensional \CAs\ this can be
shown by reversibly simulating reversible Turing machines which are
computationally universal; for $d$-dimensional \CAs\ the result holds
too (more details can be found in \cite{Kari_2005_RCA_ip}).

Hence any construction which transforms a reversible synchronous \CA\
into an invertible asynchronous \CA\ is sufficient to show the
computational universality of the latter class. It turns out that
Nakamura's method \cite{Nakamura_1974_ACA_ar} of transforming the
local transition function of any synchronous \CA\ into one for an
asynchronous \CA\ while ``basically preserving its global behavior''
(irrespective of reversibility) is all that is needed for purely
asynchronous updating.

In general the following transformations which maintain local
synchronicity and guarantee invertibility are performed:
\begin{enumerate}
  \item Each cell, additionally to its current state, remembers its
    previous state and manages a three-valued time stamp.
  \item Each \emph{active} cell only changes its state, if thereby no
    information that may be needed by neighboring cells is lost. This
    is the case, if neighboring cells have the same time stamp or
    are one step ahead.
  \item Moreover, each active cell maintains its state, if the local
    configuration it observes is ``illegal''.
\end{enumerate}
More specifically, let $C=(R,N,Q,\delta)$ and $G=(R,N,Q,\gamma)$ be two synchronous \CAs\
with corresponding global transition functions $\Delta$ and
$\Gamma$. The interesting case will be that they are inverse to each
other, \ie $\Delta^{-1} = \Gamma$. Without loss of generality, we
assume that the neighborhood includes $0$ and is symmetric, \ie $N=\{
-n\mid n\in N\}$.

We will now construct an \ACA\
$\bar{C}=(R,N,\bar{Q},\bar{\delta})$ from $C$. In order to save parentheses we
will occasionally write $\ell_n$ instead of $\ell(n)$ ($\forall l \in Q^N, n \in N$)
and $c_i$ instead of $c(i)$ ($\forall c \in Q^R, i \in R$).
The \CA\ $\bar{C}$ is defined as follows:
\begin{enumerate}
\item The set of states is $\bar{Q}= Q\x Q\x\{0,1,2\}$. For
  $\bar{q}=(q_1,q_2,t)\in \bar{Q}$ we denote the first component as
  $\curr(\bar{q})$, the second component as $\old(\bar{q})$ and the
  third as $\Time(\bar{q})$.
\item Given a local configuration $\bar{\ell}$ we say, that cell $0$ is
  \define{ahead} (of its neighbors) if, and only if, there is an
  $n\in N$ such that $\Time(\bar{\ell}_0) = \Time(\bar{\ell}_n)+1$
  ($\bmod 3$).
\item If in a local configuration $\bar{\ell}$ of $\bar{C}$ cell $0$ is
  \emph{not} ahead define the \define{corresponding current local
    $C$-configuration} $\curr(\bar{\ell})$ of $C$ (not $\bar{C}$\;!) as
  \begin{align*}
    \curr(\bar{\ell})_n &=
    \begin{cases}
      \curr(\bar{\ell}_n)
        & \text{ if } \Time(\bar{\ell}_n) = \Time(\bar{\ell}_0),   \\
      \old(\bar{\ell}_n)
        & \text{ if } \Time(\bar{\ell}_n) = \Time(\bar{\ell}_0)+1. \\
    \end{cases}
  \end{align*}
\item A local configuration $\bar{\ell}$ is \define{forward movable} if
  cell $0$ is not ahead and $\old(\bar{\ell}_0)= \gamma(\curr(\bar{\ell}))$.
\item The local transition function $\bar{\delta}$ is then defined by
  \begin{align*}
    \bar{\delta}(\bar{\ell}) &=
    \begin{cases}
      (\;\delta(\curr(\bar{\ell})),\;\curr(\bar{\ell}_0),\;\Time(\bar{\ell}_0)+1\;)
        & \text { if $\bar{\ell}$ is forward movable}, \\
      \bar{\ell}_0
        & \text { otherwise}. \\
    \end{cases}
  \end{align*}
\end{enumerate}
Analogously, apply the following construction to $G$, resulting in
$\bar{G}=(R,N,\bar{Q},\bar{\gamma})$.
\begin{enumerate}
\item $\bar{Q}= Q\x Q\x\{0,1,2\}$, $\curr(\bar{q})$, $\old(\bar{q})$
  and $\Time(\bar{q})$ are defined as above.
\item Given a local configuration $\bar{\ell}$ we say, that cell $0$ is
  \define{behind} (of its neighbors) if, and only if, there is an
  $n\in N$ such that $\Time(\bar{\ell}_0)= \Time(\bar{\ell}_n)-1 $
  ($\bmod 3$).
\item If in a local configuration $\bar{\ell}$ of $\bar{C}$ cell $0$
  is \emph{not} behind define the \define{corresponding old local
    $C$-configuration} $\old(\bar{\ell})$ of $G$ (not $\bar{G}$\;!) as
  \begin{align*}
    \old(\bar{\ell})_n &=
    \begin{cases}
      \old(\bar{\ell}_n)
        & \text{ if } \Time(\bar{\ell}_n) = \Time(\bar{\ell}_0),   \\
      \curr(\bar{\ell}_n)
        & \text{ if } \Time(\bar{\ell}_n) = \Time(\bar{\ell}_0)-1. \\
    \end{cases}
  \end{align*}
\item A local configuration $\bar{\ell}$ is \define{backward movable}
  if cell $0$ is not behind and $\curr(\bar{\ell}_0)=
  \delta(\old(\bar{\ell}))$.
\item The local transition function $\bar{\delta}$ is then defined by
  \begin{align*}
    \bar{\gamma}(\bar{\ell}) &=
    \begin{cases}
      (\;\old(\bar{\ell}_0),\;\gamma(\old(\bar{\ell})),\;\Time(\bar{\ell}_0)-1\;)
        & \text { if $\bar{\ell}$ is backward movable}, \\
      \bar{\ell}_0
        & \text { otherwise}.
    \end{cases}
  \end{align*}
\end{enumerate}
We now have
\begin{theorem}
  \label{thm:nakamura-preserves_invertibility}
  If $C$ and $G$ are synchronous \CAs\ which are inverse to each other,
  then $\bar{C}$ and $\bar{G}$ are purely asynchronous inverses of each
  other.
\end{theorem}

\begin{proof}
  Let $c,c'\in \bar{Q}^R$ be two arbitrary global configurations and
  \[
  D = D_{c,c'} = \{ i\in R \mid c_i \not= c'_i \}
  \]
  their difference. According to
  lemma~\ref{lem:purely-inversion-characterized} it is sufficient to
  prove
  \[
  c' = \bar{\Delta}_D(c) \iff c = \bar{\Gamma}_D(c') \text{.}
  \]
  We will prove in detail that
  \[
  c' = \bar{\Delta}_D(c) \Longrightarrow c = \bar{\Gamma}_D(c') \text{.}
  \]
  Because of the symmetry of the constructions is it not surprising
  that a proof of the inverse implication can be given analogously.

  Now let $c' = \bar{\Delta}_D(c)$. Consider an arbitrary cell $i\in D$.
  Since $c_i \not= c'_i = \bar{\delta}(c_{i+N})$ the local configuration
  $c_{i+N}$ is forward movable, because otherwise cell $i$ would
  maintain its state by definition of $\bar{\delta}$. Therefore
  \begin{align*}
    \forall n\in N \colon \Time(c_{i+n}) &\in \{\Time(c_i), \Time(c_i)+1\}, \\
    \old(c_i) &= \gamma(\curr(c_{i+N})) \text{,} \\
    \curr(c'_i)&= \delta(\curr(c_{i+N})) \text{,}  \\
    \old(c'_i) &= \curr(c_i) \text{, and}  \\
    \Time(c'_i) &= \Time(c_i)+1 \text{.}
  \end{align*}
  Now consider an arbitrary neighbor $n\in N$. Since we have assumed
  that $N$ is symmetric, cell $i$ is a neighbor of cell $i+n$.

  There are two possible cases:
  \begin{description}
  \item[Case 1:] $\Time(c_{i+n})=\Time(c_i)$.

    \begin{description}
    \item[Case 1.1:] $i+n\in D$: Then $c_{i+n}\not= c'_{i+n}$ and therefore
      \begin{align*}
        \Time(c'_{i+n}) &= \Time(c_{i+n})+1 = \Time(c_{i})+1 = \Time(c'_i) \text{ and} \\
        \curr(c_{i+N})(n) &= \curr(c_{i+n}) = \old(c'_{i+n}) = \old(c'_{i+N})(n) \text{.}
      \end{align*}
    \item[Case 1.2:] $i+n\not\in D$: Then $c_{i+n} = c'_{i+n}$ and therefore
      \begin{align*}
        \Time(c'_{i+n}) &= \Time(c_{i+n}) = \Time(c_{i}) = \Time(c'_i)-1 \text{ and} \\
        \curr(c_{i+N})(n) &= \curr(c_{i+n}) = \curr(c'_{i+n}) = \old(c'_{i+N})(n) \text{.}
      \end{align*}
    \end{description}
  \item[Case 2:] $\Time(c_{i+n})=\Time(c_i)+1$. Since $i$ is a
    neighbor of $i+n$ and $\Time(c_i)=\Time(c_{i+n})-1$ we have
    $c'_{i+n}=c_{i+n}$. Therefore
    \begin{align*}
      \Time(c'_{i+n}) &= \Time(c_{i+n}) = \Time(c_{i})+1 = \Time(c'_i) \text{ and} \\
      \curr(c_{i+N})(n) &= \old(c_{i+n}) = \old(c'_{i+n}) = \old(c'_{i+N})(n) \text{.}
    \end{align*}
  \end{description}
  Taken together we always have
  \begin{align*}
    \Time(c'_{i+n}) &\in \{\Time(c'_i), \Time(c'_i)-1\} \text{ and} \\
    \curr(c_{i+N})(n) &= \old(c'_{i+N})(n) \text{,}
  \end{align*}
  and therefore $\curr(c'_i) = \delta(\curr(c_{i+N})) =
  \delta(\old(c'_{i+N})) $. As a consequence $c'_{i+N}$ is backward
  movable and hence
  \begin{align*}
    \curr(c_i) &= \old(c'_i) = curr( \Gamma_D(c')_i ) \text{,} \\
    \old(c_i)  &= \gamma(\curr(c_{i+N})) = \gamma(\old(c'_{i+N})) = \old(\Gamma_D(c')_i) \text{, and} \\
    \Time(c_i) &= \Time(c'_i)-1 = \Time(\Gamma_D(c')_i) \text{,}
  \end{align*}
  which is a long-winded way of saying $c_i = \Gamma_D(c')_i $.
\end{proof}
The above proof is incorrect if we restrict purely \ACAs\ to
non-empty sets of active cells. The problematic case happens when the
minimal difference $D$ is empty. In this case we cannot use it as the
set of active cells. But nevertheless $c' = c$ may be reachable from
$c$ by $\bar{C}$ in one step with a non-empty set of active cells and
we need to prove that the same holds for $\bar{G}$. We do not know
whether this is always the case.

One part of the problem is that even
if $c$ is a global configuration in which the registers old, current and time
do not conform we need to show the property. This could be solved by
adding more restrictions that do not corrupt the simulation as in the
second part of the definitions of backward and forward movable.

\section{Decidability}
\label{sec:decidability}

We will show that phase space invertibility is decidable for
arbitrary-dimensional purely asynchronous \CAs\ and one-dimensional
fully asynchronous \CAs\ by presenting two algorithms which always
terminate and for any given automaton find an inverse if one exists.

For any given \ACA\ $C$ these algorithms only look for inverses among
the finitely many automata with the same neighborhood as $C$. This is
justified in section~\ref{subsec:neighborhood} where we prove that if
$C$ is invertible there is an inverse with the same neighborhood.

To decide whether two \ACAs\ $C$ and $G$ are inverse to each other the
algorithms only verify that $C$ and $G$ are inverse to each other on
a subspace of the phase space that is restricted to global transitions
in which only cells of a fixed finite subset of cells are active. In
sections~\ref{subsec:deciding-purely} and~\ref{subsec:deciding-fully}
we prove that this is sufficient for arbitrary-dimensional purely and
one-dimensional fully asynchronous \CAs\ respectively.

\subsection{Inverse Neighborhood and Translation Invariance}
\label{subsec:neighborhood}

Consider any set of all global configurations that agree on cell $0$
and all its neighbors. For each of these global configurations make
the global transition step where only cell $0$ is active. All
resulting global configurations again agree on cell $0$ and all its
neighbors.

An inverse can undo all these transitions, whereby at most cell $0$ is
active. Regardless of the states of cells besides cell $0$ and its
neighborhood, the inverse does the same and can thus not benefit from a
larger neighborhood.

\begin{definition}
  We say that a \CA\ $C = (R, N, Q, \delta)$ has \define{minimal
    neighborhood} if it has no dummy neighbors, meaning for each
  neighbor $n \in N$ there exist local configurations $\ell, \ell' \in
  Q^N$ with $\ell|_{N \setminus \set{n}} = \ell'|_{N \setminus \set{n}}$
  and $\ell_{n} \neq \ell'_{n}$ such that $\delta(\ell) \neq
  \delta(\ell')$.
\end{definition}
With this term we can rigorously state and prove

\begin{lemma}
  Two inverse purely or fully asynchronous \CAs\
  $C = (R, N, Q, \delta)$ and $G = (R, M, Q, \gamma)$
  with minimal neighborhoods have the same neighborhood.
\end{lemma}

\begin{proof}[by contradiction]
  Assume $N \neq M$. Without loss of generality, let $N \setminus M
  \neq \emptyset$.  Then there exists a neighbor $\mathring{n} \in N
  \setminus M$. Because the neighborhoods are minimal there also exist
  local configurations $\hat{\ell}, \check{\ell} \in Q^N$ with
  $\hat{\ell}|_{N \setminus \set{\mathring{n}}} = \check{\ell}|_{N
    \setminus \set{\mathring{n}}}$ and $\hat{\ell}_{\mathring{n}} \neq
  \check{\ell}_{\mathring{n}}$ such that $\delta(\hat{\ell}) \neq
  \delta(\check{\ell})$.

  Choose global configurations $\hat{c}, \check{c} \in Q^R$ with
  \begin{align*}
    \hat{c}_{\mathring{n}} = \hat{\ell}_{\mathring{n}} \text{ as }&\text{well as } \check{c}_{\mathring{n}} = \check{\ell}_{\mathring{n}} \text{,}\\
    \hat{c}|_{N \setminus \set{\mathring{n}}} = \hat{\ell}|_{N \setminus \set{\mathring{n}}} &= \check{\ell}|_{N \setminus \set{\mathring{n}}} = \check{c}|_{N \setminus \set{\mathring{n}}} \text{, and}\\
    \hat{c}|_{R \setminus N} &= \check{c}|_{R \setminus N} \text{.}
  \end{align*}
  For the global configurations $\hat{d} = \Delta_{\set{0}}(\hat{c})$ and $\check{d} = \Delta_{\set{0}}(\check{c})$ we obtain
  \begin{align*}
    \hat{d}_0 = \delta(\hat{\ell}) &\neq \delta(\check{\ell}) = \check{d}_0 \text{,}\\
    \hat{d}_{\mathring{n}} = \hat{\ell}_{\mathring{n}} &\neq \check{\ell}_{\mathring{n}} = \check{d}_{\mathring{n}} \text{,}\\
    \hat{d}|_{R \setminus \set{0, \mathring{n}}} &= \check{d}|_{R \setminus \set{0, \mathring{n}}} \text{, and}\\
    \hat{d}|_{R \setminus \set{0}} = \hat{c}|_{R \setminus \set{0}} \text{ as }&\text{well as } \check{d}|_{R \setminus \set{0}} = \check{c}|_{R \setminus \set{0}} \text{.}
  \end{align*}
  Because $C$ and $G$ are inverse to each other we furthermore have
  \[ \hat{c} = \Gamma_{D_{\hat{c},\hat{d}}}(\hat{d}) \text{ as well as } \check{c} = \Gamma_{D_{\check{c},\check{d}}}(\check{d}) \text{.} \]
  Note that $D_{\hat{c},\hat{d}}, D_{\check{c},\check{d}} \subseteq \set{0}$.
  See Figure~\ref{fig:case-0} for a graphical representation of the situation.
  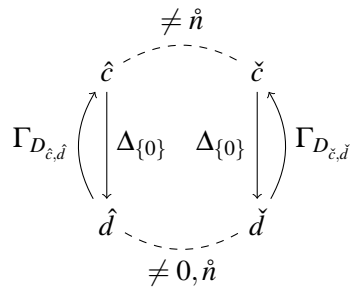
\begin{figure}[ht]
    \centering
    \begin{tikzpicture}[node distance=2cm, auto]
      \node (c1) {$\hat{c}$};
      \node (d1) [below of=c1] {$\hat{d}$};
      \node (c2) [right of=c1] {$\check{c}$};
      \node (d2) [below of=c2] {$\check{d}$};
      \draw[-, dashed, bend left] (c1) to node {$\neq \mathring{n}$} (c2);
      \draw[->] (c1) to node {$\Delta_{\set{0}}$} (d1);
      \draw[<-, bend right] (c1) to node [swap] {$\Gamma_{D_{\hat{c},\hat{d}}}$} (d1);
      \draw[->] (c2) to node [swap] {$\Delta_{\set{0}}$} (d2);
      \draw[<-, bend left] (c2) to node {$\Gamma_{D_{\check{c},\check{d}}}$} (d2);
      \draw[-, dashed, bend right] (d1) to node [swap] {$\neq 0, \mathring{n}$} (d2);
    \end{tikzpicture}
    \caption{The labels on the dashed lines indicate on which cells
      the connected global configurations differ, the solid lines
      show possible global transitions.}
    \label{fig:case-0}
  \end{figure}

  \begin{description}
    \item[Case 1:] $\mathring{n} \neq 0$.
      Then $\hat{c}_0 = \check{c}_0$. With $\hat{d}_0 \neq \check{d}_0$
      it follows that $\hat{d}_0 \neq \hat{c}_0$ or $\check{d}_0 \neq \check{c}_0$.
      Without loss of generality, let $\hat{d}_0 \neq \hat{c}_0$.
      Then $\hat{c} = \Gamma_{\set{0}}(\hat{d})$ and thus $\hat{c}_0 = \gamma(\hat{d}_{0+M})$.

      Consider global configuration $\mathring{d} \in Q^R$ with
      $\mathring{d}|_{R \setminus \set{0}} = \check{d}|_{R \setminus \set{0}}$
      and $\mathring{d}_0 = \hat{d}_0$. From $\mathring{n} \notin M$ we
      have $\mathring{d}_{0+M} = \hat{d}_{0+M}$ and therefore
      \begin{align*}
        &\Gamma_{\set{0}}(\mathring{d})_0 = \gamma(\mathring{d}_{0+M}) = \gamma(\hat{d}_{0+M}) = \hat{c}_0 = \check{c}_0 \text{ and}\\
        &\Gamma_{\set{0}}(\mathring{d})|_{R \setminus \set{0}} = \mathring{d}|_{R \setminus \set{0}} = \check{d}|_{R \setminus \set{0}} = \check{c}|_{R \setminus \set{0}} \text{,}
      \end{align*}
      which is a long-winded way of saying $\check{c} = \Gamma_{\set{0}}(\mathring{d})$.

      From $\mathring{d}_0 = \hat{d}_0 \neq \hat{c}_0 = \check{c}_0$
      we conclude $\mathring{d} = \Delta_{\set{0}}(\check{c}) = \check{d}$
      and $\check{d}_0 = \mathring{d}_0 = \hat{d}_0$, which contradicts
      $\hat{d}_0 \neq \check{d}_0$; see Figure~\ref{fig:case-1}.
      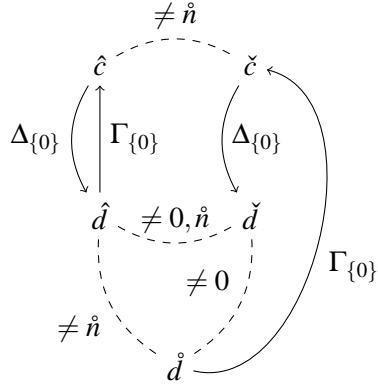
\begin{figure}[ht]
        \centering
        \begin{tikzpicture}[node distance=2cm, auto]
          \node (c1) {$\hat{c}$};
          \node (d1) [below of=c1] {$\hat{d}$};
          \node (c2) [right of=c1] {$\check{c}$};
          \node (d2) [below of=c2] {$\check{d}$};
          \node (d3) [below of=d1] at ($(d1)!.5!(d2)$) {$\mathring{d}$};
          \draw[-, dashed, bend left] (c1) to node {$\neq \mathring{n}$} (c2);
          \draw[->, bend right] (c1) to node [swap] {$\Delta_{\set{0}}$} (d1);
          \draw[<-] (c1) to node {$\Gamma_{\set{0}}$} (d1);
          \draw[->, bend right] (c2) to node {$\Delta_{\set{0}}$} (d2);
          \draw[-, dashed, bend right] (d1) to node {$\neq 0, \mathring{n}$} (d2);
          \draw[-, dashed, bend right] (d1) to node [swap] {$\neq \mathring{n}$} (d3);
          \draw[-, dashed, bend left] (d2) to node [swap] {$\neq 0$} (d3);
          \draw[->, bend right=90] (d3) to node [swap] {$\Gamma_{\set{0}}$} (c2);
        \end{tikzpicture}
        \caption{The labels on the dashed lines indicate on which cells
          the connected global configurations differ, the solid lines
          show possible global transitions.}
        \label{fig:case-1}
      \end{figure}
    \item[Case 2:] $\mathring{n} = 0$. Then $0 \notin M$ and hence
      $\hat{d}_{0+M} = \check{d}_{0+M}$ as well as $\hat{c}_{0+M} = \check{c}_{0+M}$.
      \begin{description}
        \item[Case 2.1:] $\hat{c}_0 \neq \hat{d}_0$ or $\check{c}_0 \neq \check{d}_0$.
          Without loss of generality, let $\hat{c}_0 \neq \hat{d}_0$. Then
          $\hat{c} = \Gamma_{\set{0}}(\hat{d})$ and therefore $\hat{c}_0 = \gamma(\hat{d}_{0+M})$.
          With $\hat{d}_{0+M} = \check{d}_{0+M}$ it follows that
          $\hat{c}_0 = \gamma(\check{d}_{0+M})$ and therefore
          $\check{c} \neq \hat{c} = \Gamma_{\set{0}}(\check{d})$. Hence
          $\check{c} = \Gamma_\emptyset(\check{d}) = \check{d}$ and
          $\hat{c} = \Gamma_{\set{0}}(\check{c})$.

          Finally, from $\hat{c}_0 \neq \check{c}_0$ we have
          $\check{c} = \Delta_{\set{0}}(\hat{c})$ and conclude
          $\check{d} = \hat{d}$, which contradicts $\hat{d}_0 \neq
          \check{d}_0$; see Figure~\ref{fig:case-2}.
          \begin{figure}[ht]
            \centering
            \begin{tikzpicture}[node distance=2cm, auto]
              \node (c1) {$\hat{c}$};
              \node (d1) [below of=c1] {$\hat{d}$};
              \node (c2) [right of=c1] {$\check{c}$};
              \node (d2) [below of=c2] {$\check{d}$};
              \draw[-, dashed, bend right] (c1) to node {$\neq 0$} (c2);
              \draw[-, dashed, bend right] (d1) to node {$\neq 0$} (d2);
              \draw[->, bend right] (c1) to node [swap] {$\Delta_{\set{0}}$} (d1);
              \draw[<-] (c1) to node {$\Gamma_{\set{0}}$} (d1);
              \draw[->, bend left=90] (c2) to node {$\Delta_{\set{0}}$} (d2);
              \draw[rounded corners] (c2.north west) rectangle (d2.south east);
              \draw[->, bend right] (c2) to node [swap] {$\Gamma_{\set{0}}$} (c1);
            \end{tikzpicture}
            \caption{The labels on the dashed lines indicate on which
              cells the connected global configurations differ, the
              solid lines show possible global transitions, and the
              rounded rectangle signifies that the enclosed global
              configurations are found to be identical during the
              proof.}
            \label{fig:case-2}
          \end{figure}
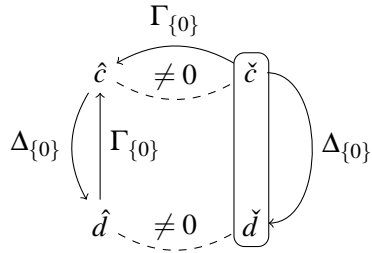

        \item[Case 2.2:] $\hat{c}_0 = \hat{d}_0$ and $\check{c}_0 = \check{d}_0$.
          Then $\hat{c} = \hat{d}$ and $\check{c} = \check{d}$. Consider
          global configuration $b = \Gamma_{\set{0}}(\hat{c}) = \Gamma_{\set{0}}(\check{c})$.
          From $\hat{c}_0 \neq \check{c}_0$ it follows that $b_0 \neq \hat{c}_0$
          or $b_0 \neq \check{c}_0$.

          Without loss of generality, let $b_0 \neq \hat{c}_0$. Then
          $\hat{c} = \Delta_{\set{0}}(b)$ and with $\check{c} \neq \hat{c}$
          also $\check{c} = \Delta_\emptyset(b) = b$. Therefore
          $\hat{c} = \Delta_{\set{0}}(\check{c}) = \check{d} = \check{c}$,
          which contradicts $\hat{c}_0 \neq \check{c}_0$; see Figure~\ref{fig:case-3}.
          \begin{figure}[ht]
            \centering
            \begin{tikzpicture}[node distance=2cm, auto]
              \node (c1) {$\hat{c}$} edge [out=210, in=240, loop] node[left] {$\Delta_{\set{0}}$} (c1);
              \node (c2) [right of=c1] {$\check{c}$} edge [out=330, in=300, loop] node[right] {$\Delta_{\set{0}}$} (c2);
              \node (b) [above of=c1] at ($(c1)!.5!(c2)$) {$b$};
              \draw[-, dashed, bend right] (c1) to node {$\neq 0$} (c2);
              \draw[->, bend left=60] (c1) to node {$\Gamma_{\set{0}}$} (b);
              \draw[<-] (c1) to node [near end, swap] {$\Delta_{\set{0}}$} (b);
              \draw[->, bend right=60] (c2) to node [swap] {$\Gamma_{\set{0}}$} (b);
            \end{tikzpicture}
            \caption{The labels on the dashed lines indicate on which
              cells the connected global configurations differ, the
              solid lines show possible global transitions.}
            \label{fig:case-3}
          \end{figure}
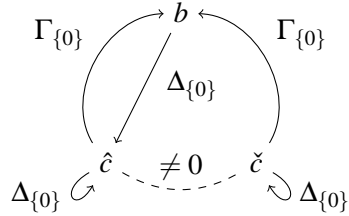
        \end{description}
  \end{description}
  Because every possible case led to a contradiction, the assumption
  must be false, meaning $N = M$.
\end{proof}
Since dummy neighbors can be added and removed without affecting the
phase space we get

\begin{corollary}
  For each invertible \ACA\ exists an inverse with the same
  neighborhood.
  \qed
\end{corollary}
Let us now briefly consider translation invariance.

\begin{definition}
  For each \define{translation vector} $j \in R$ the map
  $\tau_j: Q^R \to Q^R$ with
  \[ \forall i \in R\colon \tau_j(c)_i = c_{i+j} \]
  is called \define{($j$-)translation}.
\end{definition}
Similar to synchronous, asynchronous \CAs\ are also translation
invariant -- we only need to translate the set of active cells as well.

\begin{lemma}
  For each translation vector $j \in R$ and set of active cells
  $A \subseteq R$ holds
  \[ \tau_j \circ \Delta_A = \Delta_{A - j} \circ \tau_j \text{,} \]
  where $A - j = \set{a - j \mid a \in A}$.
  \qed
\end{lemma}

\subsection{Purely Asynchronous Cellular Automata}
\label{subsec:deciding-purely}

If every transition of a purely asynchronous \CA, in
which only cell $0$ is active as well as arbitrary neighbors of cell
$0$, can be inverted by another automaton and vice versa, then these
automata are inverse to each other.

This follows from locality properties and translation invariance of
purely asynchronous \CAs\ and is rigorously stated and proved in

\begin{lemma}
  \label{lem:restrict-active-cells}
  Two purely asynchronous \CAs\ $C = (R, N, Q, \delta)$ and
  $G = (R, N, Q, \gamma)$ are inverse to each other if, and only if,
  for each pair of global configurations $c, c' \in Q^R$ with
  $0 \in D_{c,c'} \subseteq \{0\} \cup N$ holds
  \begin{equation*}
    c' = \Delta_{D_{c,c'}}(c)
    \iff
    c = \Gamma_{D_{c,c'}}(c') \text{.}
  \end{equation*}
\end{lemma}

\begin{proof}
  The forward direction follows directly from
  lemma~\ref{lem:purely-inversion-characterized}. For the backward
  direction consider any global configurations $c, c' \in Q^R$.

  First, let $c' = \Delta_{D_{c,c'}}(c)$. Choose any cell
  $i \in D_{c,c'}$ that changes its state during the transition from $c$
  to $c'$. Because $C$ and $G$ are translation invariant it follows that
  \begin{align*}
    \Gamma_{D_{c,c'}}(c')_i
      &= \tau_i(\Gamma_{D_{c,c'}}(c'))_0\\
      &= \Gamma_{D_{c,c'}-i}(\tau_i(c'))_0\\
      &= \Gamma_{D_{c,c'}-i}(\tau_i(\Delta_{D_{c,c'}}(c)))_0\\
      &= \Gamma_{D_{c,c'}-i}(\Delta_{D_{c,c'}-i}(\tau_i(c)))_0 \text{.}
  \end{align*}
  With $0 \in D_{c,c'} - i$ we further get
  \[ \Gamma_{D_{c,c'}}(c')_i = \Gamma_{\set{0}}(\Delta_{D_{c,c'}-i}(\tau_i(c)))_0 \text{.} \]
  Because the transition state of cell $0$ only depends on its
  neighborhood we farther deduce
  \[ \Gamma_{D_{c,c'}}(c')_i = \Gamma_{\set{0}}(\underbrace{\Delta_{(D_{c,c'}-i) \cap (\set{0} \cup N)}(\overbrace{\tau_i(c)}^{d})}_{d'})_0 \text{.} \]
  We can make cells, which do not change state if active, inactive
  without disturbing the transition from $d$ to $d'$:
  \[ d' = \Delta_{D_{d,d'}}(d) \text{.} \]
  The difference $D_{d,d'}$ of $d$ and $d'$ is contained in
  $\set{0} \cup N$. From $0 \in (D_{c,c'}-i) \cap (\set{0} \cup N)$ we
  furthermore get
  \[ d_0 = \tau_i(c)_0 = c_i \neq c'_i = \delta(c_{i+N}) = \delta(\tau_i(c)_{0+N}) = d'_0 \text{,} \]
  meaning $0 \in D_{d,d'}$. Thus, by the premise of the backward
  direction,
  \[ d = \Gamma_{D_{d,d'}}(d') \]
  and therefore
  \[ \Gamma_{D_{c,c'}}(c')_i = \Gamma_{\set{0}}(d')_0 = \Gamma_{D_{d,d'}}(d')_0 = d_0 = \tau_i(c)_0 = c_i \text{.} \]
  Because cell $i$ was arbitrarily chosen, $c = \Gamma_{D_{c,c'}}(c')$
  and we have proved the implication
  \[ c' = \Delta_{D_{c,c'}}(c) \Longrightarrow c = \Gamma_{D_{c,c'}}(c') \text{.} \]
  By symmetry we conclude that the implication
  \[ c' = \Delta_{D_{c,c'}}(c) \Longleftarrow c = \Gamma_{D_{c,c'}}(c') \]
  holds as well. Because the global configurations $c$ and $c'$ were arbitrarily
  chosen $C$ and $G$ are inverse to each other according to
  lemma~\ref{lem:purely-inversion-characterized}.
\end{proof}
Thus only cell $0$, the neighbors of cell $0$, and the neighbors of each
neighbor of cell $0$ have to be considered when testing whether two
purely asynchronous \CAs\ are inverse to each other, because the
transition state of cell $0$ or an arbitrary neighbor of cell $0$ only
depends in its observed local configuration.

\begin{theorem}
  Phase space invertibility is decidable for purely asynchronous \CAs.
\end{theorem}

\begin{proof}
  Let $C=(R,N,Q,\delta)$ be an arbitrary purely asynchronous \CA.
  Consider each of the $\card{Q}^{\card{Q}^{\card{N}}}$ purely
  asynchronous \CAs\ $G=(R,N,Q,\gamma)$ in turn. Test whether
  for each pair $c, c' \in Q^{\set{0} \cup N \cup (N+N)}$ of finite
  configurations with $0 \in D_{c,c'} \subseteq \set{0} \cup N$ the
  equivalence
  \begin{equation*}
    c' = \Delta_{D_{c,c'}}(c)
    \iff
    c = \Gamma_{D_{c,c'}}(c')
  \end{equation*}
  holds. If this is the case for one automaton $G$, then $C$ and $G$
  are inverse to each other, and $C$ is invertible. Otherwise $C$ is
  not invertible.

  Note that $N+N = \set{m + n \mid m, n \in N}$.
\end{proof}

\begin{remark}
  It is in fact not necessary to consider all purely asynchronous
  \CAs\ with the same neighborhood. One can show that if $C$ is
  invertible, the local transition function of its inverse can be
  constructed easily from the one of $C$. Thus only one candidate needs
  to be considered.

  For further details see \cite[section~3.5]{Wacker_2012_RAZ_mt}.
\end{remark}
The above proofs are incorrect if we restrict purely asynchronous
\CAs\ to non-empty sets of active cells. But at least in the
one-dimensional case invertibility is still decidable (see
\cite[section~6.2.3]{Wacker_2012_RAZ_mt} for further details). This is
due to the fact that a purely asynchronous \CA\ that is restricted to
non-empty sets of active cells is invertible if, and only if, it is
invertible without the restriction and invertible as fully
asynchronous \CA\ (see
\cite[corollary~3.15]{Wacker_2012_RAZ_mt}). And, as we shall see, the
latter is decidable in the one-dimensional case.

\subsection{Fully Asynchronous Cellular Automata}
\label{subsec:deciding-fully}

In each transition of fully asynchronous \CAs\ exactly one cell is
active, whereas in each transition of purely asynchronous \CAs\
arbitrary cells are active. One may naively think that deciding
invertibility is thus easier for fully than for purely asynchronous
\CAs. But the added difficulty is that in every transition one cell
\emph{must} be active.

To decide whether two fully asynchronous \CAs\ $C$ and $G$ are inverse
to each other or not, we have to show that if $C$ transits a global
configuration $c$ into $c'$ when cell $a$ is active, that $G$ can
transit $c'$ into $c$ when a cell $a'$ is active or that no such cell
exists, and vice versa.

If $c_a \neq c'_a$ then $a$ and $a'$ must be the same cell or $C$ and
$G$ are not inverse. But if $c_a = c'_a$ then $a$ and $a'$ are in
general different cells. So how or where can we either find such a cell
$a'$ or decide that no such cell exists?

We answer this question in the one-dimensional case in

\begin{lemma}
  \label{lem:restrict-active-cells-fully}
  Let $C = (\Z, N, Q, \delta)$ and $G = (\Z, N, Q, \gamma)$ be two
  fully asynchronous one-dimensional \CAs. Define the maximal distance
  of a neighbor
  \begin{equation*}
    m = \begin{cases}
      \max_{n \in N} \abs{n} & \text{ if $N \neq \es$,}\\
      -\infty                & \text{ if $N = \es$,}
    \end{cases}
  \end{equation*}
  and a finite set of active candidate cells
  \begin{equation*}
    \mathcal{A} = \begin{cases}
      \set{-{\card{Q}}^{2m+1},\dotsc,-1,0,1,\dotsc,{\card{Q}}^{2m+1}}
              & \text{ if $m \neq -\infty$,}\\
      \set{0} & \text{ if $m = -\infty$.}
    \end{cases}
  \end{equation*}
  Then $C$ and $G$ are inverse to each other if, and only if,
  for each pair of global configurations $c, c' \in Q^{\Z}$ with
  $D_{c,c'} = \set{0}$
  \begin{equation*}
    \label{eq:lem_decide-fully-asynchron-1}
    c' = \Delta_{\set{0}}(c)
    \iff
    c = \Gamma_{\set{0}}(c')
    \tag{1} \\
  \end{equation*}
  and for each pair of global configurations $c, c' \in Q^{\Z}$ with
  $c = \Delta_{\set{0}}(c)$ and $c' = \Gamma_{\set{0}}(c')$
  \begin{equation*}
    \label{eq:lem_decide-fully-asynchron-2}
    \exists a \in \mathcal{A}: c = \Gamma_{\set{a}}(c)
    \text{ and }
    \exists a' \in \mathcal{A}: c' = \Delta_{\set{a'}}(c')
    \text{.}
    \tag{2} \\
  \end{equation*}
\end{lemma}
In the proof of this lemma we need segments of global configurations:
For each global configuration $c \in Q^{\Z}$ and cells $i, j \in \Z$
with $i \leq j$ the \define{$(i,j)$-segment (of $c$)} $c[i,j]$ is the
($j-i+1$)-tuple with $c[i,j] = (c_i,c_{i+1},\dotsc,c_{j})$.

\begin{proof}
  If there is only one state, meaning $\card{Q} = 1$, there is nothing
  to show. We thus assume $\card{Q} \geq 2$.
  \begin{itemize}
    \item[``$\Rightarrow$''] Let $C$ and $G$ be inverse to each other.
      According to lemma~\ref{lem:fully-inversion-characterized}:
      Equality~\eqref{eq:lem_decide-fully-asynchron-1} holds. To
      show equality~\eqref{eq:lem_decide-fully-asynchron-2} we choose
      an arbitrary global configuration $c \in Q^{\Z}$ with
      $c = \Delta_{\set{0}}(c)$ or $c = \Gamma_{\set{0}}(c)$.

      At first, let $N = \es$. Then $\mathcal{A} = \set{0}$. Let
      $c = \Delta_{\set{0}}(c)$. Assume $c \neq \Gamma_{\set{0}}(c)$.
      Consider the global configuration $c' \in Q^{\Z}$ with
      $c'_i = c_0$ for each cell $i \in \Z$. Then
      $c' = \Delta_{\set{0}}(c')$ and from
      \begin{equation*}
        c'_a = c_0 \neq \Gamma_{\set{0}}(c)_0
                      = \Gamma_{\set{a}}(c')_a
      \end{equation*}
      we obtain $c' \neq \Gamma_{\set{a}}(c')$ for each cell
      $a \in \Z$ -- in contradiction to the premise that $C$ and
      $G$ are inverse to each other. Thus our assumption is false,
      meaning $c = \Gamma_{\set{0}}(c)$. One can analogously show that
      $c = \Gamma_{\set{0}}(c)$ implies $c = \Delta_{\set{0}}(c)$.

      Now let $N \neq \es$. The idea is to construct a global
      configuration $c' \in Q^{\Z}$ for which
      \[ c' \in \Delta_1(c') \iff c' \in \Gamma_1(c') \]
      does \emph{not} hold if there does \emph{not} exist a cell
      $a \in \mathcal{A}$ such that $c = \Gamma_{\set{a}}(c)$, if
      $c = \Delta_{\set{0}}(c)$, or $c = \Delta_{\set{a}}(c)$, if
      $c = \Gamma_{\set{0}}(c)$, contradicting that $C$ and $G$ are
      inverse to each other.

      Consider the ${\card{Q}}^{2m+1} + 1$ segments
      \[ c[\alpha-m,\alpha+m] \]
      for cells $\alpha \in \set{0,1,\dotsc,{\card{Q}}^{2m+1}}$.
      Because at most ${\card{Q}}^{2m+1}$ of these segments are
      different, there exist two cells
      $\alpha, \alpha' \in \set{0,1,\dotsc,{\card{Q}}^{2m+1}}$ with
      $\alpha < \alpha'$ such that
      \[ c[\alpha-m,\alpha+m] = c[\alpha'-m,\alpha'+m] \text{.} \]
      Analogously there exist two cells
      $\beta', \beta \in \set{-{\card{Q}}^{2m+1},\dotsc,1,0}$ with
      $\beta' < \beta$ such that
      \[ c[\beta'-m,\beta'+m] = c[\beta-m,\beta+m] \text{.} \]

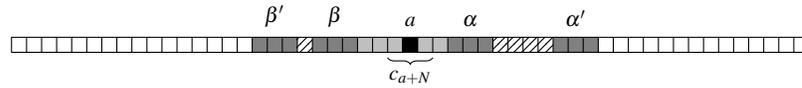
      \begin{center}
        \tikz[scale=0.2]{
          \path[fill=gray] (19,0) rectangle +(3,1) node [above,pos=0.5,yshift=3pt,scale=1] {\scriptsize$\beta'$}; 
          \path[pattern=north east lines] (22,0) rectangle +(1,1); 
          \path[fill=gray] (23,0) rectangle +(3,1) node [above,pos=0.5,yshift=3pt,scale=1] {\scriptsize$\beta$}; 
          \path[fill=lightgray] (26,0) rectangle +(3,1); 
          \path[fill] (29,0) rectangle +(1,1) node [above,pos=0.5,yshift=3pt,scale=1] {\scriptsize$a$}; 
          \draw[decoration={brace,mirror,raise=2pt}, decorate] (28,0) -- (31,0) node [below,pos=0.5,anchor=north,yshift=-3pt,scale=1] {\scriptsize$c_{a+N}$};
          \path[fill=lightgray] (30,0) rectangle +(2,1); 
          \path[fill=gray] (32,0) rectangle +(3,1) node [above,pos=0.5,yshift=3pt,scale=1] {\scriptsize$\alpha$}; 
          \path[pattern=north east lines] (35,0) rectangle +(4,1); 
          \path[fill=gray] (39,0) rectangle +(3,1) node [above,pos=0.5,yshift=3pt,scale=1] {\scriptsize$\alpha'$}; 
          \draw[step=1] (3,0) grid +(53,1);
        }
      \end{center}

      Define for each index $k \in \N$
      \begin{align*}
         \alpha_k &= \alpha'+(k-1)(\alpha'-\alpha)-m \text{,}\\
        \alpha'_k &= \alpha'+k(\alpha'-\alpha)+m \text{,}
      \end{align*}
      and
      \begin{align*}
         \beta_k &= \beta'-(k-1)(\beta-\beta')+m \text{,}\\
        \beta'_k &= \beta'-k(\beta-\beta')-m \text{.}
      \end{align*}

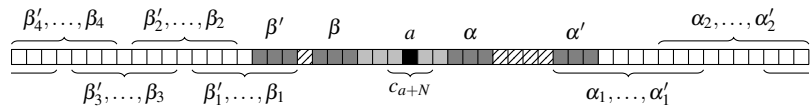
      \begin{center}
        \tikz[scale=0.2]{
          \draw[decoration={brace,raise=2pt}, decorate] (3,1) -- (10,1) node [above,pos=0.5,anchor=south,yshift=3pt,scale=1] {\scriptsize$\beta_4',\dotsc,\beta_4$};
          \draw[decoration={brace,mirror,raise=2pt}, decorate] (7,0) -- (14,0) node [below,pos=0.5,anchor=north,yshift=-3pt,scale=1] {\scriptsize$\beta_3',\dotsc,\beta_3$};
          \draw[decoration={brace,raise=2pt}, decorate] (11,1) -- (18,1) node [above,pos=0.5,anchor=south,yshift=3pt,scale=1] {\scriptsize$\beta_2',\dotsc,\beta_2$};
          \draw[decoration={brace,mirror,raise=2pt}, decorate] (15,0) -- (22,0) node [below,pos=0.5,anchor=north,yshift=-3pt,scale=1] {\scriptsize$\beta_1',\dotsc,\beta_1$};
          \path[fill=gray] (19,0) rectangle +(3,1) node [above,pos=0.5,yshift=3pt,scale=1] {\scriptsize$\beta'$}; 
          \path[pattern=north east lines] (22,0) rectangle +(1,1); 
          \path[fill=gray] (23,0) rectangle +(3,1) node [above,pos=0.5,yshift=3pt,scale=1] {\scriptsize$\beta$}; 
          \path[fill=lightgray] (26,0) rectangle +(3,1); 
          \path[fill] (29,0) rectangle +(1,1) node [above,pos=0.5,yshift=3pt,scale=1] {\scriptsize$a$}; 
          \draw[decoration={brace,mirror,raise=2pt}, decorate] (28,0) -- (31,0) node [below,pos=0.5,anchor=north,yshift=-3pt,scale=1] {\scriptsize$c_{a+N}$};
          \path[fill=lightgray] (30,0) rectangle +(2,1); 
          \path[fill=gray] (32,0) rectangle +(3,1) node [above,pos=0.5,yshift=3pt,scale=1] {\scriptsize$\alpha$}; 
          \path[pattern=north east lines] (35,0) rectangle +(4,1); 
          \path[fill=gray] (39,0) rectangle +(3,1) node [above,pos=0.5,yshift=3pt,scale=1] {\scriptsize$\alpha'$}; 
          \draw[decoration={brace,mirror,raise=2pt}, decorate] (39,0) -- (49,0) node [below,pos=0.5,anchor=north,yshift=-3pt,scale=1] {\scriptsize$\alpha_1,\dotsc,\alpha_1'$};
          \draw[decoration={brace,raise=2pt}, decorate] (46,1) -- (56,1) node [above,pos=0.5,anchor=south,yshift=3pt,scale=1] {\scriptsize$\alpha_2,\dotsc,\alpha_2'$};
          \draw[step=1] (3,0) grid +(53,1);
          \clip (3,-5) rectangle +(53,10);
          \draw[decoration={brace,mirror,raise=2pt}, decorate] (53,0) -- (63,0) node [below,pos=0.5,anchor=north,yshift=-3pt,scale=1] {}; 
          \draw[decoration={brace,mirror,raise=2pt}, decorate] (-1,0) -- (6,0) node [below,pos=0.5,anchor=north,yshift=-3pt,scale=1] {}; 
        }
      \end{center}

      Because $0 \leq \alpha < \alpha'$ and
      $\beta' < \beta \leq 0$ we have
      \begin{align*}
                                           -m < \alpha' - m = \alpha_1 &< \alpha_2 < \alpha_3 < \dotsb \text{,}\\
        m < \alpha' + m < \alpha' + (\alpha' - \alpha) + m = \alpha_1' &< \alpha_2' < \alpha_3' < \dotsb \text{,}
      \end{align*}
      and
      \begin{align*}
          \dotsb < \beta_3 < \beta_2 &< \beta_1 = \beta' + m < m \text{,}\\
        \dotsb < \beta_3' < \beta_2' &< \beta_1' = \beta' - (\beta - \beta') - m < \beta' - m < -m \text{.}
      \end{align*}

      Consider the global configuration $c' \in Q^{\Z}$ with
      \begin{align*}
        c'[\beta'-m,\alpha'+m] &= c[\beta'-m,\alpha'+m] \text{,}\\
        c'[\alpha_k,\alpha'_k] &= c[\alpha-m,\alpha'+m] \text{ ($\forall k \in \N$), and}\\
          c'[\beta'_k,\beta_k] &= c[\beta'-m,\beta+m] \text{ ($\forall k \in \N$).}
      \end{align*}

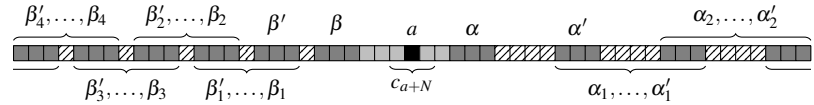
      \begin{center}
        \tikz[scale=0.2]{
          \draw[decoration={brace,raise=2pt}, decorate] (3,1) -- (10,1) node [above,pos=0.5,anchor=south,yshift=3pt,scale=1] {\scriptsize$\beta_4',\dotsc,\beta_4$};
          \draw[decoration={brace,mirror,raise=2pt}, decorate] (7,0) -- (14,0) node [below,pos=0.5,anchor=north,yshift=-3pt,scale=1] {\scriptsize$\beta_3',\dotsc,\beta_3$};
          \draw[decoration={brace,raise=2pt}, decorate] (11,1) -- (18,1) node [above,pos=0.5,anchor=south,yshift=3pt,scale=1] {\scriptsize$\beta_2',\dotsc,\beta_2$};
          \draw[decoration={brace,mirror,raise=2pt}, decorate] (15,0) -- (22,0) node [below,pos=0.5,anchor=north,yshift=-3pt,scale=1] {\scriptsize$\beta_1',\dotsc,\beta_1$};
          \begin{scope}[shift={(-16,0)}]
            \path[fill=gray] (19,0) rectangle +(3,1); 
            \path[pattern=north east lines] (22,0) rectangle +(1,1); 
          \end{scope}
          \begin{scope}[shift={(-12,0)}]
            \path[fill=gray] (19,0) rectangle +(3,1); 
            \path[pattern=north east lines] (22,0) rectangle +(1,1); 
          \end{scope}
          \begin{scope}[shift={(-8,0)}]
            \path[fill=gray] (19,0) rectangle +(3,1); 
            \path[pattern=north east lines] (22,0) rectangle +(1,1); 
          \end{scope}
          \begin{scope}[shift={(-4,0)}]
            \path[fill=gray] (19,0) rectangle +(3,1); 
            \path[pattern=north east lines] (22,0) rectangle +(1,1); 
          \end{scope}
          \path[fill=gray] (19,0) rectangle +(3,1) node [above,pos=0.5,yshift=3pt,scale=1] {\scriptsize$\beta'$}; 
          \path[pattern=north east lines] (22,0) rectangle +(1,1); 
          \path[fill=gray] (23,0) rectangle +(3,1) node [above,pos=0.5,yshift=3pt,scale=1] {\scriptsize$\beta$}; 
          \path[fill=lightgray] (26,0) rectangle +(3,1); 
          \path[fill] (29,0) rectangle +(1,1) node [above,pos=0.5,yshift=3pt,scale=1] {\scriptsize$a$}; 
          \draw[decoration={brace,mirror,raise=2pt}, decorate] (28,0) -- (31,0) node [below,pos=0.5,anchor=north,yshift=-3pt,scale=1] {\scriptsize$c_{a+N}$};
          \path[fill=lightgray] (30,0) rectangle +(2,1); 
          \path[fill=gray] (32,0) rectangle +(3,1) node [above,pos=0.5,yshift=3pt,scale=1] {\scriptsize$\alpha$}; 
          \path[pattern=north east lines] (35,0) rectangle +(4,1); 
          \path[fill=gray] (39,0) rectangle +(3,1) node [above,pos=0.5,yshift=3pt,scale=1] {\scriptsize$\alpha'$}; 
          \begin{scope}[shift={(7,0)}]
            \path[pattern=north east lines] (35,0) rectangle +(4,1); 
            \path[fill=gray] (39,0) rectangle +(3,1); 
          \end{scope}
          \begin{scope}[shift={(14,0)}]
            \path[pattern=north east lines] (35,0) rectangle +(4,1); 
            \path[fill=gray] (39,0) rectangle +(3,1); 
          \end{scope}
          \draw[decoration={brace,mirror,raise=2pt}, decorate] (39,0) -- (49,0) node [below,pos=0.5,anchor=north,yshift=-3pt,scale=1] {\scriptsize$\alpha_1,\dotsc,\alpha_1'$};
          \draw[decoration={brace,raise=2pt}, decorate] (46,1) -- (56,1) node [above,pos=0.5,anchor=south,yshift=3pt,scale=1] {\scriptsize$\alpha_2,\dotsc,\alpha_2'$};
          \draw[step=1] (3,0) grid +(53,1);
          \clip (3,-5) rectangle +(53,10);
          \draw[decoration={brace,mirror,raise=2pt}, decorate] (53,0) -- (63,0) node [below,pos=0.5,anchor=north,yshift=-3pt,scale=1] {}; 
          \draw[decoration={brace,mirror,raise=2pt}, decorate] (-1,0) -- (6,0) node [below,pos=0.5,anchor=north,yshift=-3pt,scale=1] {}; 
        }
      \end{center}

      This global configuration exists, because the rear and front parts in
      which the definitions of two consecutive indices overlap are
      precisely the identical segments
      $c[\alpha-m,\alpha+m] = c[\alpha'-m,\alpha'+m]$ and
      $c[\beta'-m,\beta'+m] = c[\beta-m,\beta+m]$. In fact, for each
      index $k \in \N$ it holds that
      \begin{align*}
        \set{\beta'-m,\dotsc,\alpha'+m}
                \cap \set{\alpha_1,\dotsc,\alpha'_1}
            &= \set{\alpha'-m,\dotsc,\alpha'+m} \text{,}\\
        \set{\beta'-m,\dotsc,\alpha'+m}
                \cap \set{\beta'_1,\dotsc,\beta_1}
            &= \set{\beta'-m,\dotsc,\beta'+m} \text{,}\\
        \set{\alpha_k,\dotsc,\alpha'_k}
                \cap \set{\alpha_{k+1},\dotsc,\alpha'_{k+1}}
            &= \set{\alpha_{k+1},\dotsc,\alpha'_k}\\
            &= \set{\alpha'_k - 2m,\dotsc,\alpha'_k}\\
            &= \set{\alpha_{k+1},\dotsc,\alpha_{k+1} + 2m} \text{,}\\
        \set{\beta'_k,\dotsc,\beta_k}
                \cap \set{\beta'_{k+1},\dotsc,\beta_{k+1}}
            &= \set{\beta'_k,\dotsc,\beta_{k+1}}\\
            &= \set{\beta'_k,\dotsc,\beta'_k + 2m}\\
            &= \set{\beta_{k+1} - 2m,\dotsc,\beta_{k+1}} \text{,}
      \end{align*}
      and
      \begin{align*}
        c'[\alpha'-m,\alpha'+m]
          &= c[\alpha'-m,\alpha'+m]\\
          &= c[\alpha-m,\alpha+m] = c'[\alpha_1,\alpha_1+2m] \text{,}\\
        c'[\beta'-m,\beta'+m]
          &= c[\beta'-m,\beta'+m]\\
          &= c[\beta-m,\beta+m] = c'[\beta_1-2m,\beta_1] \text{,}\\
        c'[\alpha_k'-2m,\alpha_k']
          &= c[\alpha'-m,\alpha'+m]\\
          &= c[\alpha-m,\alpha+m] = c'[\alpha_{k+1},\alpha_{k+1}+2m] \text{,}\\
        c'[\beta_k',\beta_k'+2m]
          &= c[\beta'-m,\beta'+m]\\
          &= c[\beta-m,\beta+m] = c'[\beta_{k+1}-2m,\beta_{k+1}] \text{.}
      \end{align*}

      We will now show that for the global configuration $c'$ any
      cell whose neighborhood is not contained in $\mathcal{A}$
      behaves the same as at least one cell whose neighborhood is
      contained in $\mathcal{A}$ if active. Note that
      $N \subseteq \set{-m,\dotsc,m}$.

      \begin{assertion}
        For each cell $a' \in \Z$ exists a cell
        $a \in \set{\beta',\dotsc,\alpha'}$ such that
        \[ c'[a'-m,a'+m] = c[a-m,a+m] \text{.} \]
      \end{assertion}

      \begin{proof}
        Let $a' \in \Z$ be an arbitrary cell. If
        $a' \in \set{\beta',\dotsc,\alpha'}$, choose
        $a = a'$. From now on let
        $a' \notin \set{\beta',\dotsc,\alpha'}$. We consider only the
        case that $a' \geq \alpha' + 1$. The case $a' \leq \beta' - 1$
        can be shown analogously. Thus let $a' \geq \alpha' + 1$.

        Because $a' \geq \alpha_1 + m$ and each of the intersections
        $\set{\alpha_k,\dotsc,\alpha'_k} \cap \set{\alpha_{k+1},\dotsc,\alpha'_{k+1}}$ ($\forall k \in \N$) and $\set{a'-m,\dotsc,a+m}$ contains $2m + 1$ elements, there exists an
        index $k \in \N$ such that
        \[ \set{a'-m,\dotsc,a'+m} \subseteq \set{\alpha_k,\dotsc,\alpha'_k} \text{.} \]
        Choose $a = a' - k(\alpha' - \alpha) \in \set{\alpha,\dotsc,\alpha'}$. Then
        \[ c'[a'-m,a'+m] = c[a-m,a+m] \text{.} \]
      \end{proof}

      With this setup we can conclude the proof:

      \begin{enumerate}
        \item Let $c = \Delta_{\set{0}}(c)$. From
          \[ c'[-m,m] = c[-m,m] \]
          we have $c' = \Delta_{\set{0}}(c')$. Assume that for each
          cell $a \in \mathcal{A}$ it holds that
          $c \neq \Gamma_{\set{a}}(c)$. Choose an arbitrary cell
          $a' \in \Z$. Then there is a cell
          $a \in \set{\beta',\dotsc,\alpha'}$ such that
          \[ c'[a'-m,a'+m] = c[a-m,a+m] \text{.} \]
          Because cell $a$ is contained in $\mathcal{A}$ it holds that
          \[ c'_{a'} = c_a \neq \Gamma_{\set{a}}(c)_a = \Gamma_{\set{a'}}(c')_{a'} \]
          and thus $c' \neq \Gamma_{\set{a'}}(c')$. Because $a' \in \Z$
          was arbitrarily chosen, this contradicts the premise that
          $C$ and $G$ are inverse to each other. Thus our assumption
          is false, meaning there exists a cell $a \in \mathcal{A}$ with
          $c = \Gamma_{\set{a}}(c)$.
        \item Let $c = \Gamma_{\set{0}}(c)$. Like above one shows that
          there exists a cell $a' \in \mathcal{A}$ with
          $c = \Delta_{\set{a'}}(c)$.
      \end{enumerate}
      Therefore equality~\eqref{eq:lem_decide-fully-asynchron-2} holds
      as well.
    \item[``$\Leftarrow$''] For each pair of global configurations $c, c'\in Q^R$ with
        $D_{c,c'} = \set{a}$ holds
        \begin{align*}
          c' = \Delta_{\set{a}}(c)
            &\iff \tau_{a}(c') = \tau_{a}(\Delta_{\set{a}}(c)) = \Delta_{\set{0}}(\tau_{a}(c)) \\
            &\overset{eq.~\eqref{eq:lem_decide-fully-asynchron-1}}{\iff} \tau_{a}(c) = \Gamma_{\set{0}}(\tau_{a}(c')) = \tau_{a}(\Gamma_{\set{a}}(c'))\\
            &\iff c = \Gamma_{\set{a}}(c') \text{,}
        \end{align*}
        and for each global configuration $c \in Q^R$ holds
        \begin{align*}
          \exists& a \in R: c = \Delta_{\set{a}}(c)\\
            &\Longrightarrow \tau_{a}(c) = \tau_{a}(\Delta_{\set{a}}(c)) = \Delta_{\set{0}}(\tau_{a}(c))\\
            &\overset{eq.~\eqref{eq:lem_decide-fully-asynchron-2}}{\Longrightarrow} \exists a' \in R: \tau_{a}(c) = \Gamma_{\set{a'}}(\tau_{a}(c)) = \tau_{a}(\Gamma_{\set{a'+a}}(c))\\
            &\Longrightarrow c = \Gamma_{\set{a'+a}}(c) \text{,}
        \end{align*}
        and analogously
        \[ \exists a \in R: c = \Gamma_{\set{a}}(c) \Longrightarrow \exists a' \in R: c = \Delta_{\set{a'+a}}(c) \text{.} \]

        Hence, by lemma~\ref{lem:fully-inversion-characterized}, the
        \CAs\ $C$ and $G$ are inverse to each other.
  \end{itemize}
\end{proof}
Thus only cell $0$, the neighbors of cell $0$, the cells in
$\mathcal{A}$ and the neighbors of each cell in $\mathcal{A}$ have to
be considered when testing whether two fully asynchronous
one-dimensional \CAs\ are inverse to each other.

\begin{theorem}
  Phase space invertibility is decidable for fully asynchronous
  one-dimensional \CAs.
\end{theorem}

\begin{proof}
  Let $C=(\Z,N,Q,\delta)$ be an arbitrary fully asynchronous
  one-dimensional \CA\ and
  $T = \set{0} \cup N \cup \mathcal{A} \cup (\mathcal{A}+N)$.
  Consider each of the
  $\card{Q}^{\card{Q}^{\card{N}}}$ fully asynchronous one-dimensional
  cellular automata $G=(\Z,N,Q,\gamma)$ in turn. Test whether
  for each pair
  $c, c' \in Q^T$ of
  finite configurations with $D_{c,c'} = \set{0}$ the equivalence
  \begin{equation*}
    c' = \Delta_{\set{0}}(c)
    \iff
    c = \Gamma_{\set{0}}(c')
  \end{equation*}
  holds and for each finite configuration
  $c \in Q^T$ the implications
  \begin{align*}
    c = \Delta_{\set{0}}(c)
      &\Longrightarrow \exists a \in \mathcal{A}: c = \Gamma_{\set{a}}(c) \text{ and}\\
    c = \Gamma_{\set{0}}(c)
      &\Longrightarrow \exists a \in \mathcal{A}: c = \Delta_{\set{a}}(c)
  \end{align*}
  hold. If this is the case for one automaton $G$, then $C$ and $G$
  are inverse to each other, and $C$ is invertible. Otherwise $C$ is
  not invertible.

  Note that $\mathcal{A}+N = \set{a + n \mid a \in \mathcal{A}, n \in N}$.
\end{proof}

\begin{remark}
  It is in fact not necessary to consider all fully asynchronous
  one-dimensional cellular automata with the same neighborhood and
  states. One can show that if $C$ is invertible, the local transition
  function of its inverse can be constructed easily from the one of $C$.
  Thus only one candidate needs to be considered.

  For further details see \cite[section~3.5]{Wacker_2012_RAZ_mt}.
\end{remark}

\section{Elementary Cellular Automata}
\label{sec:elementary}

Using the decision algorithms it can easily be shown that exactly the
purely asynchronous elementary cellular automata with Wolfram numbers
$0$, $35$, $43$, $49$, $51$, $59$, $113$, $115$, $204$, and $255$ are
invertible and that exactly the fully asynchronous elementary cellular
automata with Wolfram numbers $33$, $35$, $38$, $41$, $43$, $46$,
$49$, $51$, $52$, $54$, $57$, $59$, $60$, $62$, $97$, $99$, $102$,
$105$, $107$, $108$, $113$, $115$, $116$, $118$, $121$, $123$, $131$,
$139$, $145$, $147$, $150$, $153$, $155$, $156$, $195$, $198$, $201$,
$204$, $209$, and $211$ are invertible.

For further details see \cite[chapter~7]{Wacker_2012_RAZ_mt}.

\section{Summary and Outlook}
\label{sec:summary-and-outlook}

We have introduced a definition of invertibility for \ACAs, namely
phase space invertibility, and shown that invertible purely
asynchronous \CAs\ are computationally universal, that invertibility
can be decided for arbitrary-dimensional purely and
one-dimensional fully asynchronous \CAs.

It remains open whether invertible fully asynchronous \CAs\ are
computationally universal and whether invertibility for
higher-dimensional fully asynchronous \CAs\ is decidable. If the
latter can be shown to be true then invertibility of purely
asynchronous \CAs\ restricted to non-empty sets of active cells would
also be decidable in higher dimensions.


\end{document}